\documentclass{siamltex}

\usepackage{amsfonts, amsmath, amssymb, psfrag, graphicx, bbm, mathrsfs}

\usepackage[colorlinks=true]{hyperref}
\hypersetup{urlcolor=blue, citecolor=blue, linkcolor=blue}

\newtheorem{assumption}{Assumption}

\def\fin{\ifmmode{\Large$\diamond$}\else{\unskip\nobreak\hfil
    \penalty50\hskip1em\null\nobreak\hfil{\Large$\diamond$}
    \parfillskip=0pt\finalhyphendemerits=0\endgraf}\fi}

\def\be#1#2\ee{\begin{equation}\label{eq:#1}#2\end{equation}}
\def\req#1{{\rm(\ref{eq:#1})}}
\def\bdm  {\begin{displaymath}}
  \def\edm  {\end{displaymath}}
\def\bdmal{\begin{displaymath}\begin{aligned}}
    \def\edmal{\end{aligned}\end{displaymath}}

\mathchardef\PhiG="0108

\newcommand{\A}{{\cal A}}
\renewcommand{\L}{{\mathscr L}}
\newcommand{\N}{{\mathord{\mathbb N}}}
\newcommand{\R}{{\mathord{\mathbb R}}}

\newcommand{\B}{{\cal B}}

\renewcommand{\L}{{\mathscr L}}

\newcommand{\norm}[1]{\|#1\|}

\newcommand{\rmd}{\,\mathrm{d}}

\newcommand{\rmi}{\mathrm{i}}
\newcommand{\eps}{\varepsilon}

\def\req#1{{\rm(\ref{eq:#1})}}

\makeatletter
\newcommand{\dupdots}{\mathinner{\mkern1mu\raise\p@
    \vbox{\kern7\p@\hbox{.}}\mkern2mu
    \raise4\p@\hbox{.}\mkern2mu\raise7\p@\hbox{.}\mkern1mu}}
\makeatother
\makeatletter
\newcount\c@MaxMatrixCols \c@MaxMatrixCols=10

\makeatother

\newcommand{\gG}{\mathcal{G}}
\newcommand{\gC}{\mathcal{C}}
\newcommand{\gT}{\mathcal{T}}
\newcommand{\gF}{\mathcal{F}}
\newcommand{\gZ}{\mathcal{Z}}
\newcommand{\gZNx}{\mathcal{Z}_N^\times}

\newcommand{\gCC}{\mathfrak{C}}
\newcommand{\gTT}{\mathfrak{T}}
\newcommand{\gFF}{\mathfrak{F}}
\newcommand{\gZZ}{\mathfrak{Z}}
\newcommand{\gZZNx}{\mathfrak{Z}_N^{\times}}
 {\endMakeFramed}

\newcommand{\fbhat}{{\widehat{|f|}}}
\newcommand{\Ghat}{{\widehat G}}
\newcommand{\utilde}{{\widetilde u}}

\newcommand{\dtilde}{{\widetilde d}}
\newcommand{\ktilde}{{\widetilde k}}

\newcommand{\phitilde}{{\widetilde\varphi}}
\newcommand{\rhotilde}{{\widetilde\rho}}
\newcommand{\omtilde}{{\widetilde\omega}}
\newcommand{\V}{{{\mathscr V}_u}}

\newcommand{\Y}{{\mathscr Y}}
\newcommand{\I}{{\cal I}}
\newcommand{\J}{{\cal J}}
\newcommand{\K}{{\cal K}}
\newcommand{\RR}{{\boldsymbol{R}}}

\newcommand{\dxi}{\,\rmd \xi}
\newcommand{\dr}{\,\rmd r}
\newcommand{\dR}{\,\rmd\!R}
\newcommand{\dRR}{\,\rmd\!\RR}

\newcommand{\uu}{{u_*}}
\newcommand{\uo}{{u^*}}
\newcommand{\rs}{s}
\newcommand{\cbeta}{c_\beta}

\newcommand{\Cbeta}{C_\beta}

\makeatletter
\renewcommand\@biblabel[1]{#1.}
\makeatother

\title{Fr\'echet differentiability of molecular distribution functions II.
       The Ursell function}
\author{Martin Hanke\thanks{Institut f\"ur Mathematik, Johannes
    Gutenberg-Universit\"at Mainz, 55099 Mainz, Germany
    ({\tt hanke@math.uni-mainz.de}). The research leading to this work
    has been done within the 
    Collaborative Research Center TRR 146; corresponding funding 
    by the DFG is gratefully acknowledged.}}

\begin{document}
\sloppy
\maketitle

\begin{abstract}
For a grand canonical ensemble of classical point-like particles at equilibrium
in continuous space we investigate the functional relationship between 
a stable and regular pair potential describing the interaction of the 
particles and the thermodynamical limit of the Ursell or pair correlation 
function. 
For certain admissible perturbations of the pair potential and sufficiently 
small activity we rigorously establish Frechet differentiability 
of the Ursell function in the $L^1$ norm.

Furthermore, concerning the thermodynamical limit of the pair distribution
function we explicitly compute its Fr\'echet derivative 
as a sum of a multiplication operator and an integral operator.
\end{abstract}


\begin{AMS}
  {\sc 82B21, 82B80}
\end{AMS}

\hspace*{-0.7em}
{\footnotesize \textbf{Last modified.} \today}

\pagestyle{myheadings}
\thispagestyle{plain}
\markboth{M. HANKE}
{FRECHET DIFFERENTIABILITY OF MOLECULAR DISTRIBUTION FUNCTIONS}


\section{Introduction}
\label{Sec:Introduction}
We study a continuous system of identical classical particles in a 
grand canonical ensemble, where the potential energy is determined
by a pair potential which only depends on the distance of the interacting
particles. In the first part of this work~\cite{Hank16a} we have shown 
that in the thermodynamical limit the corresponding equilibrium  
molecular distribution functions are differentiable in $L^\infty$
with respect to the pair potential. It is well-known, however, that the
correlations between individual observations of particles become small as the 
distance between the observation points gets large. For example, the so-called
\emph{pair correlation function} or (second order) \emph{Ursell function}, 
which describes the correlations
between the occurrence of particles at two different points in space
is known to be close to zero (no correlation) for distant points,
and the rate of decay is strong enough to guarantee that the
thermodynamical limit of the pair correlation function is integrable 
over the entire space.
We mention in passing that the 
pair correlation function is important for physical chemistry applications 
(cf., e.g., R\"uhle et al.~\cite{RJLKA09})
because this is a measurable structural quantity that gives
insight into the type of underlying potential.
 
One may question whether the Fr\'echet derivative of the 
pair correlation function 
(with respect to the potential) maps also continuously into the space of
integrable functions; this does not follow from the $L^\infty$ analysis
of the first part of this work and, in fact, it does not seem possible to 
prove this with the techniques utilized in \cite{Hank16a}.
We therefore use a different argument in this paper based on 
cluster expansions of the Ursell functions. 

The same approach is subsequently
used to derive integral operator representations of the Fr\'echet derivatives
of the thermodynamical limits of the singlet and pair molecular distribution
functions, which can easily be reassembled to obtain the
derivative of the pair correlation function, when necessary. 
Among other applications such a representation may open a door to 
investigate invertibility of these derivatives.

The outline of this paper is as follows. In the following section we state
the basic assumptions on the pair potential and its perturbations, and
briefly review the main results from \cite{Hank16a}. 
Then, in Section~\ref{Sec:Cluster} we summarize classical results about cluster
expansions of grand canonical quantities such as the molecular distribution
functions and the Ursell functions; here we also recollect basic
properties of the pair correlation function, such as integrability and 
asymptotic behavior at infinity. Section~\ref{Sec:corrfcts} is devoted to
upper bounds for certain higher order correlation functions; this section 
can be skipped by readers who are only interested in our main result
on differentiability, stated and proved
in Section~\ref{Sec:L1derivative} (Theorem~\ref{Thm:Henderson-derivative}).
The estimates from Section~\ref{Sec:corrfcts} are revisited in 
Section~\ref{Sec:operator}, where they are utilized to justify the explicit
computation of the thermodynamical limit representation of the 
Fr\'echet derivative of the pair distribution function.

\section{Background}
\label{Sec:Background}
Let $\Lambda\subset\R^3$ be a bounded cubical box centered at the origin,
and $R_i\in\Lambda$, $i=1,2,\dots$, be the coordinates of the individual 
particles of a grand canonical ensemble in $\Lambda$. Repeatedly we use 
the notation
\bdm
   \RR_N = (R_1,\dots,R_N) \qquad \text{and} \qquad
   \RR_{n,N} = (R_{n+1},\dots,R_N)
\edm
for the coordinates of (some of) the particles of the entire ensemble.
When the system is in thermal equilibrium the $m$ particle distribution 
function given by
\be{rho-m}
   \rho^{(m)}_\Lambda(\RR_m)
   \,=\, \frac{1}{\Xi_\Lambda}
         \sum_{N=m}^\infty \frac{z^N}{(N-m)!}
         \int_{\Lambda^{N-m}}\!\! 
            e^{-\beta U_N(\RR_N)}
         \dRR_{m,N}
\ee
describes -- up to proper normalization --
the probability density of observing $m$ particles simultaneously
at the coordinates $R_1,R_2,\dots,R_m\in\Lambda$. 
In \req{rho-m} 
\be{pairpotential}
   U_N(R_1,\dots,R_N) \,=\, \!\!\sum_{1\leq i<j\leq N} u(|R_i-R_j|)
\ee
is the potential energy of a configuration of an $N$ particle system,
assuming that the interactions between the particles can be described 
by a pair potential
and that these interactions only depend on their mutual distances.
Furthermore, in \req{rho-m} $\beta>0$ is the inverse temperature, 
$z>0$ is the \emph{activity}, and
\be{Xi}
   \Xi_\Lambda \,=\, 
   \sum_{N=0}^\infty \frac{z^N}{N!}\int_{\Lambda^N} e^{-\beta U_N(\RR_N)}\dRR_N
\ee
is the associated \emph{grand canonical partition function}. 

Following \cite{Hank16a} we declare the pair potential $u:\R^+\to\R$ 
to satisfy the following assumption.

\begin{assumption}
\label{Ass:u}
There exists $\rs>0$ and positive decreasing functions $\uu,\uo:\R^+\to\R$ 
with
\begin{align*}
   \int_0^\rs \uu(r)\,r^2\!\dr \,=\, \infty \qquad &\text{and} \qquad 
   \int_\rs^\infty \uo(r)\,r^2\!\dr \,<\, \infty\,,
\intertext{such that $u$ satisfies}
   u(r) \,\geq\, \uu(r)\,, \ \ r \,\leq\, \rs\,,\qquad &\text{and} \qquad
   |u(r)| \,\leq\, \uo(r)\,, \ \ r \,\geq\, \rs\,.
\end{align*}
\end{assumption}

We also introduce the Banach space $\V$ of \emph{perturbations}
of $u$ as the set of functions $v$ for which the corresponding norm
\be{V}
   \norm{v}_{\V} \,=\, 
   \max\{\,\norm{v/u}_{(0,s)},\norm{v/\uo}_{(s,\infty)}\,\}
\ee
is finite\footnote{If $\Omega\subset\R^d$ is a domain then 
$\norm{\,\cdot\,}_\Omega$ denotes the supremum norm over $\Omega$.}. 
With these prerequisites it has been shown in \cite{Hank16a} that 
for any $0<t_0<1$ the following three properties hold true for all 
perturbed potentials $\utilde=u+v$ with $v\in\V$, $\norm{v}_\V\leq t_0$,
the respective quantities being independent of the particular choice of $v$:
\begin{itemize}
\item[(i)] there exists $B>0$ such that
\be{B}
   \sum_{1\leq i<j\leq N} \!\! \utilde(|R_i-R_j|) \,\geq\, -BN
\ee
for every configuration of $N$ particles and every $N\in\N$, i.e.,
$\utilde$ is \emph{stable};
\item[(ii)] for every $m\in\N$ and $\RR_m\in\Lambda^m$ there exists 
an index $j^*(\RR_m)$ such that
\be{jstar}
   \sum_{i=1\atop \,\,\,i\neq j^*}^m \utilde(|R_i-R_{j^*}|) \,\geq\, -2B
\ee
with the same constant $B$ as in \req{B};
\item[(iii)]
there exists $\cbeta>0$ with
\be{cbeta}
   4\pi\int_0^\infty |e^{-\beta \utilde(r)}-1|\,r^2\!\dr \,\leq\, \cbeta\,,
\ee
i.e., $\utilde$ is \emph{regular}.
\end{itemize}
As a consequence it follows from classical results by 
Ruelle~\cite{Ruel69} that for any of these potentials the
associated molecular distribution functions $\rhotilde_\Lambda^{(m)}$, $m\in\N$,
converge compactly to a bounded limiting function $\rhotilde^{(m)}$
as the size of $\Lambda$ grows to infinity, provided that the activity
is sufficiently small, i.e., that
\be{z}
   0 \,<\, z \,<\,\frac{1}{\cbeta e^{2\beta B+1}}\,.
\ee
This is known as the \emph{thermodynamical limit}. 

In \cite{Hank16a} it has further been shown that the molecular distribution
functions $\rho_\Lambda^{(m)}$ as well as their thermodynamical limits
have Fr\'echet derivatives
\bdm
   \partial\rho_\Lambda^{(m)}\in\L(\V,\L^\infty(\Lambda^m))
   \qquad \text{and} \qquad
   \partial\rho^{(m)}\in\L(\V,\L^\infty((\R^3)^m))
\edm
with respect to $u$,
and for a given $v\in\V$ the directional derivatives
$\bigl(\partial\rho_\Lambda^{(m)})v$ converge compactly to 
$\bigl(\partial\rho^{(m)})v$ in the thermodynamical limit, uniformly for
$\norm{v}_\V\leq 1$.

In this paper we focus on the pair correlation function
\be{Ursell2}
   \omega_\Lambda^{(2)}(R_1,R_2) 
   \,=\, \rho_\Lambda^{(2)}(R_1,R_2)
         \,-\, \rho_\Lambda^{(1)}(R_1)\rho_\Lambda^{(1)}(R_2)\,.
\ee
In the thermodynamical limit the pair correlation function converges 
(compactly) to a function $\omega^{(2)}$ that only depends on $|R_2-R_1|$, 
and which is related to the so called
\emph{radial distribution function} $g$ via
\bdm
   g(r) \,=\, 1 \,+\, \frac{1}{\rho_0^2}\,\omega^{(2)}(R,0)\,, \qquad |R|=r\,,
\edm
where $\rho_0=\lim_{|\Lambda|\to\infty}\rho_\Lambda^{(1)}$ 
is the (constant) \emph{counting density} of the system.

From our aforementioned results it follows that $\omega^{(2)}(R,0)$ as
a function of $R\in\R^3$ is Fr\'echet differentiable in 
$\L(\V,L^\infty(\R^3))$. However, this function is also known 
to belong to $L^1(\R^3)$, cf.~\cite{Ruel69}, and to converge to zero for 
$|R|\to\infty$ under mild additional assumptions on $u$, 
cf.~\cite{Penr67,Groe67,Groe67b,PoUe09} and Section~\ref{Sec:Cluster} below.
This decay at infinity is not taken into account when studying the 
distribution functions in $L^\infty$ as has been done in \cite{Hank16a};
therefore the purpose of this paper is to extend our results and to prove 
that $\omega^{(2)}(\,\cdot\,,0)$ has a Fr\'echet derivative
$\partial\omega^{(2)}\in\L(\V,L^1(\R^3))$.

Our method of proof utilizes classical graph theoretical cluster expansions
that have been developed in the aforementioned
papers to derive appropriate bounds for the pair correlation function.
We will summarize the corresponding ingredients in the following section.

\section{Cluster expansions}
\label{Sec:Cluster}
A graph $\gG$ is a set of (undirected) bonds $(i,j)$ between labeled vertices, 
where $(i,j)\in\gG$ means that there is a bond connecting vertices 
$\#i$ and $\#j$. In our applications vertex $\#i$ corresponds to the 
$i$th particle of the grand canonical ensemble and its coordinates 
$R_i\in\Lambda$; a bond $(i,j)\in\gG$ is associated with a 
certain interaction of the corresponding two particles, either given by
\be{fij}
  f_{ij} \,=\, f(R_i-R_j)\,, \qquad R_i,R_j\in\R^3\,,
\ee
where
\be{Mayer}
   f(R) \,=\, e^{-\beta u(|R|)} - 1
\ee 
is the so-called \emph{Mayer $f$-function}, or by the absolute values of 
$f_{ij}$. We refer to Stell~\cite{Stel64} as a general reference and for
a detailed exposition of graph theory in statistical mechanics.

For our results three types of graphs are relevant. First comes the set of
\emph{connected graphs}: in a connected graph
every pair of vertices has a connecting path of bonds between them. 
Connected graphs can be used to specify the sequence 
$(\omega_\Lambda^{(m)})_{m\geq 1}$ of Ursell functions, i.e.,
\be{qm-graph}
   \omega_\Lambda^{(m)}(\RR_m)
   \,=\, \sum_{N=m}^\infty \frac{z^N}{(N-m)!}\,
         \sum_{\gC_N} \int_{\Lambda^{N-m}} 
               \Bigl(\!\!\!\!\prod_{(i,j)\in\gC_N}\!\!\!\!f_{ij}\Bigr) 
               \dRR_{m,N}\,,
\ee
where the sum varies over all connected graphs $\gC_N$ with 
$N$ vertices labeled $\#1$ through $\#N$.
The second Ursell function has already been introduced in \req{Ursell2}
and the first one can be shown to coincide with $\rho_\Lambda^{(1)}$,
compare~\req{rhom-graph} below; further examples that we need later on are
\be{Ursell3}
\begin{aligned}
   \omega_\Lambda^{(3)}(R_1,R_2,R_3)
   &\,=\, \rho_\Lambda^{(3)}(R_1,R_2,R_3)
          \,-\, \rho_\Lambda^{(2)}(R_2,R_3)\rho_\Lambda^{(1)}(R_1) \\[1ex]
   &\qquad \quad
          \,-\, \omega_\Lambda^{(2)}(R_1,R_3)\rho_\Lambda^{(1)}(R_2)
          \,-\, \omega_\Lambda^{(2)}(R_1,R_2)\rho_\Lambda^{(1)}(R_3)
\end{aligned}
\ee
and
\be{Ursell4}
\begin{aligned}
   &\omega_\Lambda^{(4)}(R_1,R_2,R_3,R_4)
    \,=\, \rho_\Lambda^{(4)}(R_1,R_2,R_3,R_4) 
          \,-\, \rho_\Lambda^{(2)}(R_1,R_2)\rho_\Lambda^{(2)}(R_3,R_4)\\[1ex]
   &\qquad
          \,-\, \omega_\Lambda^{(3)}(R_1,R_2,R_3)\rho_\Lambda^{(1)}(R_4)
          \,-\, \omega_\Lambda^{(2)}(R_1,R_4)\omega_\Lambda^{(2)}(R_2,R_3)
          \\[1ex]
   &\qquad
          \,-\, \omega_\Lambda^{(3)}(R_1,R_2,R_4)\rho_\Lambda^{(1)}(R_3)
          \,-\, \omega_\Lambda^{(2)}(R_1,R_3)\omega_\Lambda^{(2)}(R_2,R_4)
          \\[1ex]
   &\qquad
          \,-\, \rho_\Lambda^{(3)}(R_1,R_3,R_4)\rho_\Lambda^{(1)}(R_2)
          \,-\, \rho_\Lambda^{(3)}(R_2,R_3,R_4)\rho_\Lambda^{(1)}(R_1)\\[1ex]
   &\qquad
          \,+\, 2\rho_\Lambda^{(2)}(R_3,R_4)\rho_\Lambda^{(1)}(R_1)
                \rho_\Lambda^{(1)}(R_2)\,.
\end{aligned}
\ee 
This last representation may not be the simplest one, but it is the one that
we will exploit below.

To introduce a second set of graphs let $\I,\J\subset\N$
be two disjoint finite sets of vertex labels with cardinalities $|\I|\geq 1$
and $|\J|\geq 0$. We define $\gZZ_{\I,\J}$ as the set of graphs with 
vertices given by $\I\cup\J$, out of which those in $\I$ are 
``highlighted'' -- being white as opposed to black, say -- 
and where each black vertex has a connecting path of bonds
to one of the white vertices. These graphs occur in the expansion 
\be{rhom-graph}
   \rho^{(m)}_\Lambda(\RR_m)
   \,=\, \sum_{N=m}^\infty \frac{z^N}{(N-m)!}\,
   \sum_{\gZ_{N,m}} \int_{\Lambda^{N-m}} 
               \Bigl(\!\!\!\!\prod_{(i,j)\in\gZ_{N,m}}\!\!\!\!\!\!f_{ij}\Bigr) 
               \dRR_{m,N}
\ee
of the molecular distribution functions, where $\gZ_{N,m}$ varies over all
graphs in $\gZZ_{\I_m,\J_{m,N}}$ with $\I_m=\{1,\dots,m\}$ and 
$\J_{m,N}=\{m+1,\dots,N\}$. 

A special case of the latter graphs are \emph{trees} and \emph{forests}.
A tree is a connected graph with a single white vertex,
its \emph{root}, such that between each pair of vertices there
is one and only one connecting path. A union of trees 
is a forest; the set of forests whose constituent trees have the same roots
$\I$ and the same black vertices $\J$ is denoted by 
$\gFF_{\I,\J}\subset\gZZ_{\I,\J}$.

Concerning trees we quote the following well-known result
(for a proof, cf., e.g., Penrose~\cite{Penr67}), which will be referred to
later on:

\begin{lemma}
\label{Lem:Penrose}
Let $u$ be a stable and regular pair potential and denote by $\gTT_N$ the
set of trees with $N$ vertices labeled $\#1$ through $\#N$.
If $z$ satisfies \req{z} then the series
\be{tau1}
   \tau_\Lambda^{(1)}(R_1)
   \,=\, \sum_{N=1}^\infty \frac{(ze^{2\beta B})^N}{(N-1)!}
            \sum_{\gT_N\in\gTT_N} \int_{\Lambda^{N-1}} 
               \Bigl(\!\!\prod_{(i,j)\in \gT_N}\!\!|f_{ij}|\Bigr)\dRR_{1,N}
\ee
converges uniformly in $\Lambda$, and there holds
\be{w}
   \norm{\tau_\Lambda^{(1)}}_\Lambda
   \,\leq\, w \,:=\, -\frac{1}{\cbeta}\,W(-z\cbeta e^{2\beta B})
   \,<\,\frac{1}{\cbeta}\,,
\ee
where $W$ is the Lambert $W$-function, cf., e.g., \cite{Lambert}.
Moreover, 
\be{tau2}
   \tau_\Lambda^{(2)}(R_1,R_2)
   \,=\, \sum_{N=2}^\infty \frac{(ze^{2\beta B})^N}{(N-2)!}
            \sum_{\gT_N\in\gTT_N} \int_{\Lambda^{N-2}} 
               \Bigl(\!\!\prod_{(i,j)\in \gT_N}\!\!|f_{ij}|\Bigr)\dRR_{2,N}
\ee
converges uniformly in $\Lambda^2$, and there holds
\bdm
   \tau_\Lambda^{(2)}(R_1,R_2) \,\leq\, G(R_1-R_2)\,, \qquad 
   R_1,R_2\in\Lambda\,,
\edm
where $G:\R^3\to\R$ is given by
\be{convint}
   G(R) \,=\, \int_{\R^3} 
                 e^{2\pi\rmi \xi\cdot R} \, 
                 \frac{w^2\,\fbhat(\xi)}{1-w\,\fbhat(\xi)}
              \dxi
\ee
and $\fbhat$ is the Fourier transform 
\bdm
   \fbhat(\xi) \,=\, 
   \int_{\R^3} e^{-2\pi\rmi\xi\cdot R}\,\bigl| e^{-\beta u(|R|)}-1 \bigr|\dR
\edm
of the absolute values of the Mayer $f$ function~\req{Mayer}.
\end{lemma}

For later convenience we list a few properties of the function $G$.

\begin{proposition}
\label{Prop:G}
Let the assumptions of Lemma~\ref{Lem:Penrose} be satisfied.
Then the function $G$ of \req{convint} is nonnegative, even, bounded, 
and integrable. Moreover, if $u(r)\to 0$ as $r\to\infty$
(e.,g., when $u$ satisfies Assumption~\ref{Ass:u})
then $G(R)\to 0$ for $|R|\to\infty$.
\end{proposition}

\begin{proof}
According fo Lemma~\ref{Lem:Penrose} $G$ majorizes the nonnegative function
$\tau^{(2)}$, hence $G$ is nonnegative and even by virtue of \req{convint}.

Since $|f|$ is bounded by $e^{2\beta B}$ and belongs to $L^1(\R)$,
cf.~\req{B} and \req{cbeta}, respectively, it follows that $f\in L^2(\R^3)$,
and then $\fbhat\in L^2(\R^3)$, too. 
Therefore, rewriting \req{convint} as
\bdm
   \Ghat \,=\, w^2\,\fbhat \,+\, \frac{w^3}{1-w\,\fbhat}\,\fbhat^2
\edm
we conclude that the second term on the right-hand side belongs to 
$L^1(\R^3)$ because its numerator is bounded away from zero according to
\req{w}, and taking the inverse Fourier transform we obtain
\bdm
   G - w^2 |f| \in C_0\,.
\edm
From this it follows that $G$ is bounded, and if $u$ vanishes at infinity 
then so does $G$.

Finally, it follows from \req{convint} and the nonnegativity of $G$ that
\bdm
   \int_{\R^3} \bigl|G(R)\bigr| \!\dR \,=\, \Ghat(0)
   \,=\, \frac{w^2\,\fbhat(0)}{1-w\,\fbhat(0)}
   \,<\, \infty\,,
\edm
i.e., $G\in L^1(\R^3)$.
\end{proof}

Let $\I,\J\subset\N$ with $\I\cap\J=\emptyset$ be given; 
furthermore, denote by $\RR_\I$ and $\RR_\J$ the coordinates of the particles 
with labels in $\I$ and $\J$, respectively.
Ruelle~\cite{Ruel64,Ruel69} considered the functions
\be{varphi-Ruelle}
   \varphi_{\I,\J}(\RR_\I;\RR_\J)
   \,=\, \sum_{\gZ_{\I,\J}}\,
         \Bigl(\!\!\!\!\prod_{(i,j)\in\gZ_{\I,\J}}\!\!\!\!\!\!f_{ij}\Bigr)\,,
\ee
where $\gZ_{\I,\J}$ varies over all graphs in $\gZZ_{\I,\J}$; for
$\I=\emptyset$ the definition \req{varphi-Ruelle} is set to be zero
because $\gZZ_{\emptyset,\J}$ is the empty set. 
Take note that the order of the particles in $\I$ and $\J$ does not affect 
the value of the right-hand side of \req{varphi-Ruelle}.

Given $i^*=i^*(\I,\RR_\I)\in\I$, 
and eliminating vertex $\#i^*$ from all graphs in $\gZ_{\I,\J}$, 
Ruelle derived the recursion
\bdm
   \varphi_{\I,\J}(\RR_\I;\RR_\J)
   \,=\, d_{\I,i^*}(\RR_\I)
         \sum_{\K\subset\J}
            k_\K(R_{i^*};\RR_\K)
            \varphi_{\I\cup\K\setminus\{i^*\},\J\setminus\K}
               (\RR_{\I\cup\K\setminus\{i^*\}};\RR_{\J\setminus\K})\,,
\edm
where 
\begin{align}
\label{eq:dm}
   d_{\I,i^*}(\RR_\I)
   &\,=\,\prod_{i\in\I\atop \,\,\,i\neq i^*}\!e^{-\beta u(|R_i-R_{i^*}|)}\\[1ex]
\intertext{and}
\label{eq:kernel}
   k_\K(R;\RR_\K)
   &\,=\, \prod_{j\in\K} f(R_j-R)\,.
\end{align}
We emphasize that the particular elements of the family of $\varphi$-functions
that enter into this recursion depend on the actual values of the input 
coordinates $\RR_\I$ because we explicitly allow $i^*$ to depend on $\RR_\I$;
aside of that the listing of the variables $\RR_\I$ and $\RR_\J$ is
redundant here and below, because it's always the coordinates of the particles
associated with the two indices of $\varphi$ that are
used as corresponding arguments. We therefore simplify our notation and
follow~\cite{PoUe09} by writing $\varphi(\I;\J)$ instead of 
$\varphi_{\I,\J}(\RR_\I;\RR_\J)$ in the remainder of this work; 
similarly we will write $d_{\I,i^*}$ and $k_\K(R_{i^*})$ for the left-hand sides
of \req{dm} and \req{kernel}, respectively. The above recursion
thus takes the form 
\be{Ruelle-recursion}
   \varphi(\I;\J)
   \,=\, d_{\I,i^*} \sum_{\K\subset\J}
            k_\K(R_{i^*})\varphi(\I\cup\K\setminus\{i^*\};\J\setminus\K)\,.
\ee

If we select $i^*=j^*(\I,\RR_\I)\in\I$ in accordance with \req{jstar} 
in such a way that
\be{jstar2}
   \sum_{i\in\I\atop \,\,\,i\neq j^*} \utilde(|R_i-R_{j^*}|) \,\geq\, -2B
\ee
for every $\utilde=u+v$ with $\norm{v}_\V\leq t_0$
then it follows from \req{Ruelle-recursion} that
\be{Ueltschi-recursion-tmp}
   \bigl|\varphi(\I;\J)\bigr|
   \,\leq\, e^{2\beta B}\! \sum_{\K\subset\J} \bigl|k_\K(R_{j^*})\bigr|
               \bigl|\varphi(\I\cup\K\setminus\{j^*\};\J\setminus\K)\bigr|\,,
\ee
and by induction Ruelle concluded that
\be{Ruelle-induction}
   \int_{\Lambda^{|\J|}} \bigl| \varphi(\I;\J)\bigr|
   \dRR_\J
   \,\leq\, (|\J|)!\,\cbeta^{|\J|}
            \bigl(e^{2\beta B+1}\bigr)^{|\I\cup\J|-1}\,.
\ee
Later, Poghosyan and Ueltschi~\cite{PoUe09} considered the functions
\be{Ueltschi}
   \psi(\I;\J)
   \,=\, e^{2(N-1)\beta B}\!\!\sum_{\gF\in\gFF_{\I,\J}}\,
         \Bigl(\!\!\!\!\prod_{(i,j)\in \gF}\!\!\!\!|f_{ij}|\Bigr)\,,
   \qquad N=|\I\cup\J|\,,
\ee
which satisfy the recursion \req{Ueltschi-recursion-tmp} with \emph{equality}, 
i.e.,
\be{Ueltschi-recursion}
   \psi(\I;\J)
   \,=\, e^{2\beta B}\! \sum_{\K\subset\J} \bigl|k_\K(R_{j^*})\bigr|
         \psi(\I\cup\K\setminus\{j^*\};\J\setminus\K)\,.
\ee
From this they readily obtained a so-called 
\emph{tree-graph inequality}\footnote{In fact,
Poghosyan and Ueltschi's normalization of the functions $\psi$ is slightly 
different and the bounds~\req{varphi-psi} in \cite{PoUe09} are off by 
$e^{2\beta B}$ for $|\I|\geq 2$ and $e^{4\beta B}$ for $|\I|=1$; 
this can be fixed by a more careful initialization of their 
inductive argument.},
namely 
\begin{subequations}
\label{eq:varphi-psi}
\begin{align}
\label{eq:varphi-psi-a}
   \bigl|\varphi(\I;\J)\bigr| 
   &\,\leq\, \psi(\I;\J)\,,\phantom{e^{-2\beta B}xxx}
   |\I| \geq 2\,, 
\intertext{and}
\label{eq:varphi-psi-b}
   \bigl|\varphi(\I;\J)\bigr| 
   &\,\leq\, e^{-2\beta B}\psi(\I;\J)\,,\phantom{xxx}
   |\I| = 1\,.
\end{align}
\end{subequations}

The set $\gCC_N$ of connected graphs
with vertices labeled $\#1$ to $\#N$ and the set $\gTT_N$ of trees with the 
same vertices agree -- up to the color of their vertices --
with the two sets $\gZZ_{\{1\},\J_{1,N}}$ and $\gFF_{\{1\},\J_{1,N}}$, respectively,
with $J_{1,N}=\{2,\dots,N\}$. 
It therefore follows from \req{varphi-Ruelle}, \req{varphi-psi-b}, 
and \req{Ueltschi} that
\bdmal
    \Biggl|\,
       \sum_{\gC_N\in\gCC_N}
       \int_{\Lambda^{N-2}}
          \Bigl(\!\!\!\!\prod_{(i,j)\in\gC_N}\!\!\!\!f_{ij}\Bigr) \!
       \dRR_{2,N}
    \Biggr|
   &\,=\, \Biggl| 
             \int_{\Lambda^{N-2}} \varphi(\{1\};\{2,\dots,N\})\!\dRR_{2,N}
          \Biggr| \\[1ex]
   &
    \,\leq\, e^{2(N-2)\beta B} 
             \sum_{\gT_N\in\gTT_N}
             \int_{\Lambda^{N-2}} 
         \Bigl(\!\!\!\!\prod_{(i,j)\in \gT_N}\!\!\!\!\!|f_{ij}|\Bigr)\!
         \dRR_{2,N} \,,
\edmal
and hence, cf.~\req{qm-graph} and Lemma~\ref{Lem:Penrose},
\be{Penrose}
\begin{aligned}
   \omega_\Lambda^{(2)}(R_1,R_2)
   &\,\leq\, e^{-4\beta B}
             \sum_{N=2}^\infty \frac{(ze^{2\beta B})^N}{(N-2)!}\,
             \sum_{\gT_N\in\gTT_N}
             \int_{\Lambda^{N-2}} 
               \Bigl(\!\!\!\!\prod_{(i,j)\in \gT_N}\!\!\!\!\!|f_{ij}|\Bigr)\!
             \dRR_{2,N}\\[1ex]
   &\,\leq\, e^{-4\beta B} G(R_1-R_2)\,.
\end{aligned}
\ee
Since the right-hand side of \req{Penrose} does not depend on $\Lambda$ 
we can turn to the thermodynamical limit to conclude that
\bdm
   \omega^{(2)}(R,0) \,\leq\, e^{-4\beta B}G(R)\,.
\edm
Accordingly, it follows from Proposition~\ref{Prop:G} that 
$\omega^{(2)}(\,\cdot\,,0)\in L^1(\R^3)$, and that $\omega^{(2)}(R,0)\to 0$ as 
$|R|\to\infty$ if $u$ satisfies Assumption~\ref{Ass:u}.

\section{Higher order correlation functions}
\label{Sec:corrfcts}
In this section we provide similar estimates for the higher order
correlation functions
\begin{align}
\label{eq:chi3}
   \chi_\Lambda^{(3)}(R_1,R_2,R_3)
   &\,=\, \rho_\Lambda^{(3)}(R_1,R_2,R_3)
          \,-\, \rho_\Lambda^{(1)}(R_1)\rho_\Lambda^{(2)}(R_2,R_3)
\intertext{and}
\label{eq:chi4}
   \chi_\Lambda^{(4)}(R_1,R_2,R_3,R_4)
   &\,=\, \rho_\Lambda^{(4)}(R_1,R_2,R_3,R_4)
          \,-\, \rho_\Lambda^{(2)}(R_1,R_2)\rho_\Lambda^{(2)}(R_3,R_4)\,,
\end{align}
which we will need in Section~\ref{Sec:operator}.
To this end we introduce another set of graphs:
For $N\geq 2$ let $\gZZNx$ be the set of graphs in $\gZZ_{\I_2,\J_{2,N}}$,
where $\I_2=\{1,2\}$ and $\J_{2,N}=\{3,\dots,N\}$,
that have no connecting edge between vertices $\#1$ and $\#2$. Moreover, let
\be{varphip}
   \varphi_N^\times(\RR_N) \,=\, \!\!\sum_{\gZNx\in\gZZNx}\,
      \Bigl(\!\!\!\prod_{(i,j)\in\gZNx}\!\!\!\!f_{ij}\Bigr)
\ee
and
\be{zetam}
   \zeta^{(m)}(\RR_m)
   \,=\, \sum_{N=m}^\infty \frac{z^N}{(N-m)!}
         \int_{\Lambda^{N-m}}
            \varphi_N^\times(\RR_N)
         \dRR_{m,N}\,.
\ee

\begin{lemma}
\label{Lem:zetam}
Let $u$ be stable and regular and $z$ satisfy \req{z}. Then there holds
\bdm
   \bigl|\zeta^{(3)}(R_1,R_2,R_3)\bigr|
   \,\leq\, 
   we^{-4\beta B} \bigl(G(R_1-R_3)+G(R_2-R_3)\bigr)\,,
\edm
where the constant $w$ and the function $G$ are defined in 
Lemma~\ref{Lem:Penrose}.
\end{lemma}

\begin{proof}
For a given value of $N\geq 2$ we adopt Ruelle's method mentioned in the 
previous section and eliminate vertex $\#1$ from each of the graphs in 
$\gZZNx$ to obtain the identity
\bdm
   \varphi_N^\times(\RR_N)
   \,=\, \!\!\sum_{\K\subset\J_{2,N}}\!\!
            k_\K(R_1) \varphi(\K\cup\{2\};\J_{2,N}\setminus\K)\,.
\edm
Note that the corresponding term $d_{\I_2,1}$ of \req{Ruelle-recursion} 
is missing here because all graphs in $\gZZNx$ lack a connecting edge 
between vertices $\#1$ and $\#2$. 
Therefore, \req{varphi-psi-a} and \req{Ueltschi-recursion} yield
\be{Lem:varphip}
   \bigl|\varphi_N^\times(\RR_N)\bigr|
   \,\leq\, \!\!\sum_{\K\subset\J_{2,N}}\!\!
            \bigl|k_\K(R_1)\bigr|
            \psi(\K\cup\{2\};\J_{2,N}\setminus\K)
   \,=\, e^{-2\beta B} \psi(\I_2;\J_{2,N})\,,
\ee
and inserting this inequality into \req{zetam} we arrive at
\bdm
   \bigl|\zeta^{(3)}(R_1,R_2,R_3)\bigr|
   \,\leq\, e^{-2\beta B}
            \sum_{N=3}^\infty \frac{z^N}{(N-3)!} 
            \int_{\Lambda^{N-3}} \psi(\I_2;\J_{2,N}) \dRR_{3,N}\,.
\edm
From the definition~\req{Ueltschi} of $\psi(\I_2;\J_{2,N})$ 
we therefore conclude that
\be{zetam-estimate}
   \bigl|\zeta^{(3)}(R_1,R_2,R_3)\bigr|
   \,\leq\, e^{-4\beta B}(S_1^{(3)}+S_2^{(3)})
\ee
with
\bdm
   S_l^{(3)} 
   \,=\, \sum_{N=3}^\infty \!\frac{(ze^{2\beta B})^N}{(N-3)!}
         \sum_{\gF_l}
         \int_{\Lambda^{N-3}} 
           \Bigl(\!\!\!\!\prod_{(i,j)\in \gF_l}\!\!\!\!|f_{ij}|\Bigr)
         \dRR_{3,N}\,, \qquad l=1,2\,,
\edm
where the inner sum varies over all those forests 
$\gF_l\in\gFF_{\I_2,\J_{2,N}}$, for which the vertices $\#3$ and $\#l$ 
belong to the same tree. 

The forests that occur in $S_1^{(3)}$ consist of all possible combinations of
one tree involving vertices $\#1$ and $\#3$ and another tree rooted in 
vertex~$\#2$;
hence, we can use Lemma~\ref{Lem:Penrose} and classical graph integral
calculus, cf.~\cite{HaMcD13,Stel64} to estimate
\bdm
   S_1^{(3)} \,\leq\, w\,G(R_1-R_3)\,.
\edm
Likewise we obtain
\bdm
   S_2^{(3)} \,\leq\, w\,G(R_2-R_3)\,.
\edm
Inserting this into \req{zetam-estimate} we thus obtain the assertion.
\end{proof}

Next we recall from \req{Ursell3} that
\bdm
   \chi_\Lambda^{(3)}(R_1,R_2,R_3)
   \,=\, \omega_\Lambda^{(3)}(R_1,R_2,R_3) 
         + \omega_\Lambda^{(2)}(R_1,R_3)\rho_\Lambda^{(1)}(R_2)
         + \omega_\Lambda^{(2)}(R_1,R_2)\rho_\Lambda^{(1)}(R_3)\,.
\edm
From this and \req{qm-graph} we conclude that $\chi_\Lambda^{(3)}$ is a sum 
over all graphs with $N\geq 3$ vertices and the associated graph integrals 
over $\RR_{3,N}\in\Lambda^{N-3}$,
where the graphs are of either one of the following three types:
\begin{itemize}
\item[(i)] a connected graph;
\item[(ii)] a graph with two connected components, one of which containing
vertex $\#2$ and the other one containing vertices $\#1$ and $\#3$;
\item[(iii)] a graph with two connected components, one of which containing
vertex $\#3$ and the other one containing vertices $\#1$ and $\#2$.
\end{itemize}
We use this observation to establish the following result.

\begin{proposition}
\label{Prop:chi3}
If $u$ is a stable and regular pair potential and the activity $z$ 
satisfies \req{z} then
\bdm
   \bigl|\chi_\Lambda^{(3)}(R_1,R_2,R_3)\bigr|
   \,\leq\, 
   we^{-4\beta B} e^{-\beta u(|R_2-R_3|)}
   \bigl(G(R_2-R_1) + G(R_3-R_1)\bigr)
\edm
for all $R_1,R_2,R_3\in\Lambda$.
\end{proposition}

\begin{proof}
By its definition $\zeta^{(3)}(R_3,R_2,R_1)$
-- note the different ordering of the arguments --
is the sum over all graphs with $N\geq 3$ vertices
and the associated graph integrals over $\RR_{3,N}$, where in each graph
every vertex is connected to vertex~$\#2$ or to vertex~$\#3$, 
but the latter two vertices have no connecting edge. 
Adding a bond between vertices $\#2$ and $\#3$ to any of these 
graphs therefore results in a connected graph. Accordingly, any graph 
occuring in the definition of $\zeta^{(3)}(R_3,R_2,R_1)$
belongs to the list (i)-(iii) above,
and so does its counterpart with the additional bond.

Likewise, if $\gC\in\gCC_N$ with $N\geq 3$,
and if one eliminates the edge between vertices $\#2$ and $\#3$ when present,
then, still, every vertex has a connecting path to vertex~$\#2$ or 
to vertex $\#3$. We therefore conclude that the 
graphs appearing in \req{chi3} consist of all those taken care of in
$\zeta^{(3)}(R_3,R_2,R_1)$ and their counterparts with an additional bond
between vertices $\#2$ and $\#3$.
Thus it follows from the definition of the graph integrals and the 
distributive law that
\bdm
   \chi_\Lambda^{(3)}(R_1,R_2,R_3)
   \,=\, e^{-\beta u(|R_2-R_3|)}\zeta^{(3)}(R_3,R_2,R_1)\,.
\edm
Now the assertion follows from Lemma~\ref{Lem:zetam}.
\end{proof}

Finally we turn to $\chi_\Lambda^{(4)}$ of \req{chi4}.
A straightforward computation based on \req{Ursell4} and \req{chi3} 
reveals that
\be{chi4-identity}
\begin{aligned}
   &\chi_\Lambda^{(4)}(R_1,R_2,R_3,R_4)
    \,=\, \eta(R_1,R_2,R_3,R_4)\\[1ex]
   &\qquad \qquad
    \,+\, \chi_\Lambda^{(3)}(R_1,R_3,R_4)\rho_\Lambda^{(1)}(R_2)
    \,+\, \chi_\Lambda^{(3)}(R_2,R_3,R_4)\rho_\Lambda^{(1)}(R_1)\,,
\end{aligned}
\ee
where
\bdmal
   &\eta(R_1,R_2,R_3,R_4)
    \,=\, \omega_\Lambda^{(4)}(R_1,R_2,R_3,R_4)
          \,+\, \omega_\Lambda^{(3)}(R_1,R_2,R_3)\rho_\Lambda^{(1)}(R_4)\\[1ex]
   &\qquad \qquad \quad \
    \,+\, \omega_\Lambda^{(2)}(R_1,R_4)\omega_\Lambda^{(2)}(R_2,R_3) 
    \,+\, \omega_\Lambda^{(2)}(R_1,R_3)\omega_\Lambda^{(2)}(R_2,R_4) \\[1ex]
   &\qquad \qquad \quad \
    \,+\, \omega_\Lambda^{(3)}(R_1,R_2,R_4)\rho_\Lambda^{(1)}(R_3)\,.
\edmal
As above we observe that $\eta$ is the sum of all graph integrals that 
correspond to connected graphs with $N\geq 4$ vertices
having a connecting edge between vertices $\#3$ and $\#4$,
and their counterparts which are obtained when deleting this very edge.
The latter ones are the graphs from $\cup_{N\geq 4}\gZZNx$ -- 
up to the labeling of the two white vertices; therefore it follows as above
that
\be{eta-zeta4}
   \eta(R_1,R_2,R_3,R_4)
   \,=\, e^{-\beta u(|R_3-R_4|)} \zeta^{(4)}(R_3,R_4,R_1,R_2)\,.
\ee

\begin{proposition}
\label{Prop:chi4}
If $u$ is a stable and regular pair potential and the activity $z$ 
satisfies \req{z} then there exists $C>0$ independent of the size of 
$\Lambda$, such that
\be{Prop:chi4}
   \bigl|\chi_\Lambda^{(4)}(R_1,R_2,R_3,R_4)\bigr|
   \,\leq\, C e^{-\beta u(|R_3-R_4|)}\sum_{i=1}^2\sum_{j=3}^4 G(R_i-R_j)
\ee
for all $R_1,R_2,R_3,R_4\in\Lambda$.
\end{proposition}

\begin{proof}
From \req{zetam}, \req{Lem:varphip}, and the definition~\req{Ueltschi} of 
$\psi(\I_2;\J_{2,N})$ we obtain as in the proof of Lemma~\ref{Lem:zetam} that
\be{zeta4estimate}
   \bigl|\zeta^{(4)}(R_1,R_2,R_3,R_4)\bigr|
   \,\leq\, e^{-4\beta B}(S_1^{(4)}+S_2^{(4)}+S_3^{(4)}+S_4^{(4)})\,,
\ee
where
\bdm
   S_l^{(4)} 
   \,=\, \sum_{N=4}^\infty \!\frac{(ze^{2\beta B})^N}{(N-4)!}
         \sum_{\gF_l}
         \int_{\Lambda^{N-4}} 
           \Bigl(\!\!\!\!\prod_{(i,j)\in \gF_l}\!\!\!\!|f_{ij}|\Bigr)
         \dRR_{4,N}\,, \qquad l=1,\dots,4\,,
\edm
where the inner sum varies over all those forests 
$\gF_l\in\gFF_{\I_2,\J_{2,N}}$, for which in case of
\begin{itemize}
\item[$l=1$:] vertex $\#3$ belongs to the tree rooted in vertex $\#1$, 
and vertex $\#4$ does not;
\item[$l=2$:] vertex $\#4$ belongs to the tree rooted in vertex $\#1$,
and vertex $\#3$ does not;
\item[$l=3$:] vertex $\#3$ and $\#4$ belong to the same tree rooted in 
vertex $\#1$;
\item[$l=4$:] vertex $\#3$ and $\#4$ belong to the same tree rooted in 
vertex $\#2$.
\end{itemize}
Standard graph analysis and Lemma~\ref{Lem:Penrose} immediately lead to 
bounds for the first two cases, namely
\bdm
   S_1^{(4)} \,\leq\, G(R_1-R_3)G(R_2-R_4)\,, \qquad
   S_2^{(4)} \,\leq\, G(R_1-R_4)G(R_2-R_3)\,.
\edm

Concerning $S_3^{(4)}$ we first consider those forests (in $S_{31}^{(4)}$, say)
where the connecting path between vertices $\#1$ and $\#4$ passes 
through vertex $\#3$, and the remaining forests (in $S_{32}^{(4)}$)
where the path between vertices $\#1$ to $\#3$ passes through vertex~$\#4$. 
In the first case the trees rooted in vertex~$\#1$ can be constructed by 
glueing together a tree rooted in vertex $\#1$ and containing vertex $\#3$ 
and a second tree rooted in vertex $\#3$ and containing vertex $\#4$;
this yields the bound
\bdm
   S_{31}^{(4)} \,\leq\, \frac{w}{ze^{2\beta B}}\,G(R_1-R_3)G(R_3-R_4)\,,
\edm
where the numerator is due to the fact that
vertex $\#3$ is a joint vertex of the two trees that are glued together, and
the extra factor $w$ stems from the tree rooted in vertex $\#2$.

Likewise, we obtain corresponding bounds for $S_{32}^{(4)}$ and $S_4^{(4)}$,
namely
\bdm
   S_{32}^{(4)} \,\leq\, \frac{w}{ze^{2\beta B}}\,G(R_1-R_4)G(R_3-R_4)
\edm
and 
\bdm
   S_4^{(4)} \,\leq\, \frac{w}{ze^{2\beta B}}\,\bigl(G(R_2-R_3) + G(R_2-R_4)\bigr)
                     G(R_3-R_4)\,.
\edm
Since $G$ is bounded, cf.~Proposition~\ref{Prop:G}, we finally obtain by
inserting all these bounds into \req{zeta4estimate} that
\bdm
   \bigl|\zeta^{(4)}(R_1,R_2,R_3,R_4)\bigr| 
   \,\leq\, C
    \bigl(G(R_1-R_3)+G(R_2-R_3)+G(R_1-R_4)+G(R_2-R_4)\bigr)\,.
\edm
Together with \req{eta-zeta4} this yields
\bdm
   \bigl|\eta(R_1,R_2,R_3,R_4)\bigr|
   \,\leq\, C e^{-\beta u(|R_3-R_4|)} \sum_{i=1}^2\sum_{j=3}^4 G(R_i-R_j)\,.
\edm
From Proposition~\ref{Prop:chi3} it follows that a similar inequality
(with a different constant) holds true for the last two terms 
of \req{chi4-identity} either, hence the proof is done.
\end{proof}

\section{Differentiability of the pair correlation function in
$\boldsymbol{L^1(\R^3)}$}
\label{Sec:L1derivative}
Throughout this section we consider perturbations $\utilde=u+v$
of a given potential $u$ that satisfies Assumption~\ref{Ass:u}, where
$v\in\V$ with $\norm{v}_\V\leq t_0/2$ is kept fixed. Associated with $u$ and
two finite index sets $\I$ and $\J$ with $\I\cap\J=\emptyset$
are the Ruelle functions $\varphi(\I;\J)$ of \req{varphi-Ruelle}, and
we will associate corresponding Ruelle functions $\phitilde(\I;\J)$ with
the perturbed potential $\utilde$. Later on we also resort to the
pair correlation function $\omtilde_\Lambda^{(2)}$ corresponding to the 
grand canonical ensemble with interaction potential $\utilde$.

We need a few auxiliary estimates from \cite{Hank16a}. The first
one, compare Lemma~3.2 in \cite{Hank16a},
concerns the functions $d_{\I,j^*}$ of \req{dm} with $j^*=j^*(\I,\RR_\I)$ 
selected as in \req{jstar2}:
If $\dtilde_{\I,j^*}$ is the corresponding function associated
with $\utilde$ and if $\norm{v}_\V\leq t_0/2$ then there holds
\begin{subequations}
\label{eq:Lem:d-derivative-all}
\begin{align}
\label{eq:dm-diff}
   \norm{\dtilde_{\I,j^*}-d_{\I,j^*}}_{(\R^3)^{|\I|}}
   &\,\leq\, \frac{2e^{2\beta B}}{t_0}\,\norm{v}_{\V}\,,\\[1ex]
\label{eq:18.x}
   \norm{\dtilde_{\I,j^*} - d_{\I,j^*} - (\partial d_{\I,j^*})v}_{(\R^3)^{|\I|}}
   &\,\leq\, \frac{4e^{2\beta B}}{t_0^2}\,\norm{v}_{\V}^2 \,,
\end{align}
\end{subequations}
where $\partial d_{\I,j^*}$ is the Fr\'echet derivative of $d_{\I,j^*}$ 
whose specific form is given in \cite{Hank16a} but is not relevant 
for our purposes below.
Take note that the estimates \req{Lem:d-derivative-all} make use of the fact
that the index $j^*(\I,\RR_\I)$ does not depend on $v$ 
because of our smallness assumption, cf.~\req{jstar2}.
Second, for $\K\subset\N$ let $k_\K$ be given by \req{kernel} and
$\ktilde_\K$ be the corresponding function associated with $\utilde$.
Then, since $\norm{v}_\V\leq t_0$, there exists a constant $\Cbeta>0$ such that
\begin{subequations}
\label{eq:kn-induction}
\begin{align}
\label{eq:kn-induction1}
   \sup_{R\in\R^3}\bigl\|\ktilde_\K(R)-k_\K(R)\bigr\|_{L^1((\R^3)^n)}
   &\,\leq\, n\Cbeta \cbeta^{n-1}\,\norm{v}_{\V}\,,\\[1ex]
\label{eq:kn-induction2}
   \sup_{R\in\R^3}
      \big\|\ktilde_\K(R) - k_\K(R) - \bigl((\partial k_\K)v\bigr)(R)
      \bigr\|_{L^1(\R^3)^n)}
   &\,\leq\, n^2 \Cbeta \cbeta^{n-1} \norm{v}^2_{\V}\,,
\end{align}
\end{subequations}
where $n=|\K|$; see the proof of Proposition~3.3 
in \cite{Hank16a}; again, the specific form of the Fr\'echet derivative
$\partial k_\K$ does not matter.

Now we can estimate the difference between the Ruelle functions
associated with $u$ and $\utilde=u+v$.

\begin{lemma}
\label{Lem:Ruelle-diff}
Under the assumptions of this section let
$\I,\J\subset\N$ be two finite index sets with $\I\neq\emptyset$ and
$\I\cap\J=\emptyset$. Then there holds
\be{Ruelle-diff-induction}
\begin{aligned}
   &\int_{\Lambda^{|\J|}}
     \Bigl|
        \phitilde(\I;\J) \,-\, \varphi(\I;\J)
     \Bigr|\dRR_\J\\[1ex]
   &\qquad \qquad \,\leq\,  
   \bigl(|\I\cup\J|-1\bigr)\,\frac{2\cbeta+t_0 \Cbeta}{t_0\cbeta}\,
   (|\J|!)\cbeta^{|\J|}
   \bigl(e^{2\beta B+1}\bigr)^{|\I\cup\J|-1}\norm{v}_{\V}\,.
\end{aligned}
\ee
\end{lemma}

\begin{proof}
Before we start we define
\be{Deltaphi}
   \Delta_0\varphi(\I;\J) \,=\, 
   \phitilde(\I;\J)\,-\,\varphi(\I;\J)\,.
\ee
Now the proof proceeds by induction on $|\I\cup\J|$. When $|\I\cup\J|=1$, i.e.,
when $\I$ consists of a single element and $\J=\emptyset$, then
\bdm
   \varphi_{\I,\J} \,=\, \phitilde_{\I,\J} \,=\, 1
\edm
by virtue of \req{varphi-Ruelle},
and hence the assertion is obviously correct. Note that
\req{Ruelle-diff-induction} is also true when 
$\I=\emptyset$ and $\J\neq\emptyset$ is arbitrary; this will be used in the
induction step.

Concerning the induction, we use \req{Ruelle-recursion} with 
$i^*=j^*(\I,\RR_\I)$ of \req{jstar2} for 
$|\I\cup\J|\geq 2$, $\I\neq\emptyset$, to derive the recursion
\bdmal
   \Delta_0\varphi(\I;\J) 
   &\,=\,
    (\dtilde_{\I,j^*}-d_{\I,j^*})
    \sum_{\K\subset\J} 
       \ktilde_\K(R_{j^*})
       \phitilde(\I\cup\K\setminus\{j^*\};\J\setminus\K)\\
   &\quad \ + 
    d_{\I,j^*} \sum_{\K\subset\J}
                \bigl(\ktilde_\K(R_{j^*})-k_\K(R_{j^*})\bigr)
                \phitilde(\I\cup\K\setminus\{j^*\};\J\setminus\K)\\
   &\quad \ + 
    d_{\I,j^*} \sum_{\K\subset\J}
                k_\K(R_{j^*})
                \Delta_0\varphi(\I\cup\K\setminus\{j^*\};\J\setminus\K)\,.
\edmal
Integrating over $\RR_\J$ and utilizing
\req{dm-diff}, \req{kn-induction1}, Ruelle's estimate \req{Ruelle-induction},
and the induction hypothesis~\req{Ruelle-diff-induction} we thus obtain 
\bdmal
   \int_{\Lambda^{|\J|}}
    \bigl|\Delta_0\varphi(\I,\J)&\bigr|\dRR_\J
    \,\leq\, \frac{1}{e}(e^{2\beta B+1})^{|\I|+|\J|-1}\cbeta^{|\J|}\,\cdot\\[1ex]
   &\sum_{\K\subset\J} (|\J|-|\K|)!\, 
    \Bigl( \frac{2}{t_0} \,+\, |\K|\,\frac{\Cbeta}{\cbeta}
           \,+\, \bigl(|\I\cup\J|-2\bigr)\frac{2\cbeta+t_0 \Cbeta}{t_0\cbeta}
    \Bigr)\norm{v}_\V \,.
\edmal
Since the right-hand side only depends on the number $p$ of elements in $\K$
we can sum over $p$ instead which gives the upper bound
\bdmal
   \int_{\Lambda^{|\J|}}
   \bigl|\Delta_0\varphi(\I;\J)\bigr|\dRR_\J
   \,\leq\, &\,\frac{1}{e}(e^{2\beta B+1})^{|\I|+|\J|-1} \cbeta^{|\J|}(|\J|!) 
             \norm{v}_\V \,\cdot\\[1ex]
   &\quad
    \sum_{p=0}^\infty \frac{1}{p!}\,
    \Bigl( \frac{2}{t_0}\,+\, p\,\frac{\Cbeta}{\cbeta}
           \,+\, \bigl(|\I\cup\J|-2\bigr)\frac{2\cbeta+t_0 \Cbeta}{t_0\cbeta}
    \Bigr)
\edmal
which coincides with \req{Ruelle-diff-induction}.
\end{proof}

Under the assumptions of Lemma~\ref{Lem:Ruelle-diff} and
with the same notation as before we let
\bdm
   \varphi'(\I;\J) \,=\, 0 \qquad \text{for} \quad |\I|=0 \quad \text{or} \quad
                           |\I|=1,\,|\J|=0\,,
\edm
and for $\I\cap\J=\emptyset$, $|\I\cup\J|\geq 2$, $|\I|\neq0$, 
and with $j^*=j^*(\I,\RR_\I)$ we define recursively
\be{varphi-prime}
\begin{aligned}
   \varphi'(\I;\J) 
   &\,=\, (\partial d_{\I,j^*})v
          \sum_{\K\subset\J} 
            k_\K(R_{j^*})
            \varphi(\I\cup\K\setminus\{j^*\};\J\setminus\K)\\
   &\quad \ + 
    d_{\I,j^*} \sum_{\K\subset\J}
                \bigl((\partial k_\K)v\bigr)(R_{j^*})
                \varphi(\I\cup\K\setminus\{j^*\};\J\setminus\K)\\
   &\quad \ + 
    d_{\I,j^*} \sum_{\K\subset\J}
                k_\K(R_{j^*})
                \varphi'(\I\cup\K\setminus\{j^*\};\J\setminus\K)\,.
\end{aligned}
\ee
Take note that $\varphi'$ depends linearly on $v$.

\begin{lemma}
\label{Lem:Ruelle-derivative}
Under the same assumptions as in Lemma~\ref{Lem:Ruelle-diff} there exists
a constant $C$ such that
\bdmal
   &\int_{\Lambda^{|\J|}}
      \Bigl|
         \phitilde(\I;\J) \,-\, \varphi(\I;\J) \,-\, \varphi'(\I;\J)
      \Bigr|\dRR_\J\\[1ex]
   &\qquad \qquad \,\leq\,  
    C (|\I|+|\J|-1)^2 (|\J|!)\,\cbeta^{|\J|}
            (e^{2\beta B+1})^{|\I|+|\J|-1}
            \norm{v}_{\V}^2\,.
\edmal
The constant $C$ only depends on $u$ and on $t_0$, cf.~\req{gamma},
but neither on the size of $\Lambda$ nor on $\RR_\I\in\Lambda^{|\I|}$.
\end{lemma}

\begin{proof}
Again the proof proceeds by induction on $|\I|+|\J|$, where for 
$|\I|+|\J|=1$ there is nothing to prove. 
Utilizing the notations~\req{Deltaphi} and 
\be{Taylor1}
   \Delta_1\varphi(\I;\J)
   \,=\, \phitilde(\I;\J)-\varphi(\I;\J)-\varphi'(\I;\J) \\[1ex]
\ee
for the zeroth and first order Taylor remainders of $\varphi(\I;\J)$
we can use the recursions \req{Ruelle-recursion}
and \req{varphi-prime} for $N:=|\I|+|\J|-2\geq 0$, $|\I|\neq 0$, to obtain
\bdmal
   &\Delta_1\varphi(\I;\J) \,=\, 
    \bigl(\dtilde_{\I,j^*}-d_{\I,j^*}-(\partial d_{\I,j^*})v\bigr)
    \sum_{\K\subset\J} k_\K(R_{j^*}) 
       \varphi(\I\cup\K\setminus\{j^*\};\J\setminus\K)\\
   &\quad \quad \ +\, 
    (\dtilde_{\I,j^*}-d_{\I,j^*})
    \sum_{\K\subset\J}\bigl(\ktilde_\K(R_{j^*})-k_\K(R_{j^*})\bigr)
       \varphi(\I\cup\K\setminus\{j^*\};\J\setminus\K)\\
   &\quad \quad \ +\, 
    (\dtilde_{\I,j^*}-d_{\I,j^*})\sum_{\K\subset\J}\ktilde_\K(R_{j^*})
       \Delta_0\varphi(\I\cup\K\setminus\{j^*\};\J\setminus\K)\\
   &\quad \quad \ +\, 
    d_{\I,j^*} \sum_{\K\subset\J}
    \bigl(\ktilde_\K(R_{j^*})-k_\K(R_{j^*})-((\partial k_\K)v)(R_{j^*})\bigr)
    \varphi(\I\cup\K\setminus\{j^*\};\J\setminus\K)\\
   &\quad \quad \ +\, 
    d_{\I,j^*} \sum_{\K\subset\J}\bigl(\ktilde_\K(R_{j^*})-k_\K(R_{j^*})\bigr)
    \Delta_0\varphi(\I\cup\K\setminus\{j^*\};\J\setminus\K)\\
   &\quad \quad \ +\, 
    d_{\I,j^*} \sum_{\K\subset\J}k_\K(R_{j^*})
    \Delta_1\varphi(\I\cup\K\setminus\{j^*\};\J\setminus\K)\,.
\edmal
Integrating over $\RR_\J\in\Lambda^{|\J|}$ and using the inequalities
\req{18.x}, \req{Ruelle-induction}, \req{dm-diff}, \req{kn-induction1},
\req{Ruelle-diff-induction}, \req{kn-induction2}, and the induction hypothesis
then we obtain in the same way as in the proof of Lemma~\ref{Lem:Ruelle-diff} 
that
\bdm
   \int_{\Lambda^{|\J|}} \bigl|\Delta_1\varphi(\I;\J)\bigr|\dRR_\J \\
   \,\leq\,
   \bigl(C \,+\, 2C N + C N^2\bigr) 
   (|\J|!)\,\cbeta^{|\J|} e^{(2\beta B+1)(N+1)} \norm{v}^2_{\V}\,,
\edm
provided that we let
\be{gamma}
   C \,=\, 
   \max\,\Bigl\{\,
\text{\footnotesize$  
                \frac{4\cbeta+2\Cbeta t_0+2\Cbeta t_0^2}
                     {t_0^2\cbeta} 
                \ , \,
                \frac{1}{2}
                \Bigl(\frac{2\cbeta+t_0\Cbeta}{t_0\cbeta}\Bigr)^2
                $}\,
         \Bigr\}\,.
\ee
Since $N+1=|\I|+|\J|-1$ the induction step is complete.
\end{proof}

Now we come to the main result of this section.

\begin{theorem}
\label{Thm:Henderson-derivative}
Let $u$ satisfy Assumption~\ref{Ass:u} and $z$ be constrained by \req{z}.
Then the thermodynamical limit $\omega^{(2)}(\,\cdot\,,0)$ 
of the corresponding pair correlation function~\req{Ursell2} is
Fr\'echet differentiable in $\L\bigl(\V,L^1(\R^3)\bigr)$.
\end{theorem}

\begin{proof}
From \cite{Hank16a} we know that the molecular
distribution functions $\rho_\Lambda^{(1)}$ and $\rho_\Lambda^{(2)}$ are 
Fr\'echet differentiable with respect to $L^\infty(\Lambda)$,
respectively $L^\infty(\Lambda^2)$. Hence, the function 
\bdm
   \omega_\Lambda^{(2)}(R_1,R_2)
   \,=\, \rho_\Lambda^{(2)}(R_1,R_2) \,-\, 
         \rho_\Lambda^{(1)}(R_1)\rho_\Lambda^{(1)}(R_2)\,,
   \qquad R_1,R_2\in\Lambda\,,
\edm
has a Fr\'echet derivative $\partial\omega_\Lambda^{(2)}$ (with respect to $u$)
in the same topology, and this implies
Fr\'echet differentiability with the same derivative also in $L^1(\Lambda^2)$. 
For a given $v\in\V$ define
\bdm
   \omega_\Lambda'(R_1,R_2) \,=\, 
   \sum_{N=2}^\infty \frac{z^N}{(N-2)!} 
      \int_{\Lambda^{N-2}} \varphi'(\{1\};\{2,\dots,N\})\dRR_{N-2}
\edm
with $\varphi'$ of \req{varphi-prime}.
Using the notation from \req{Taylor1} it follows from 
Lemma~\ref{Lem:Ruelle-derivative} that
\be{omdiff-integrability}
\begin{aligned}
   &\int_{\Lambda}\,
      \Bigl|
         \omtilde_\Lambda^{(2)}(R_1,R_2)-\omega_\Lambda^{(2)}(R_1,R_2)
         \,-\, \omega_\Lambda'(R_1,R_2)
      \Bigr|\dR_2\\[1ex]
   &\qquad
    \leq \int_\Lambda\sum_{N=2}^\infty \frac{z^N}{(N-2)!}
    \int_{\Lambda^{N-2}}
       \Bigl| \Delta_1\varphi\bigl(\{1\}\,;\{2,\dots,N\}\bigr)
       \Bigr|\dRR_{2,N}\dR_2\\[1ex]
   &\qquad
    \leq\, \sum_{N=2}^\infty
              C z^N (N-1)^3
              \bigl(\cbeta e^{2\beta B+1}\bigr)^{N-1}
              \norm{v}_{\V}^2 
    \,=\, C' \norm{v}_{\V}^2\,,
\end{aligned}
\ee
where the constant $C'$ is finite because of \req{z} and
independent of the size of $\Lambda$ and independent of the choice of
$R_1\in\Lambda$. Since $\omega_\Lambda'$ depends
linearly on $v$ this inequality reveals that
\bdm
   (\partial\omega_\Lambda^{(2)})v \,=\, \omega_\Lambda'\,.
\edm

From \cite{Hank16a} we know that
$(\partial\rho_\Lambda^{(2)})v\to(\partial\rho^{(2)})v$ and
$(\partial\rho_\Lambda^{(1)})v\to(\partial\rho^{(1)})v$ compactly as
$|\Lambda|\to\infty$. The latter is necessarily a constant 
denoted by $(\partial\rho_0)v$ in the sequel for brevity.
We also know that $\rho_\Lambda^{(1)}\to \rho_0$ compactly as 
$|\Lambda|\to\infty$. It thus follows that
\bdmal
   \bigl((\partial\omega_\Lambda^{(2)})v\bigr)(R_1,R_2) \,\to\,
   &\,\bigl((\partial\rho^{(2)})v\bigr)(R_1,R_2)
   \,-\, 2\rho_0(\partial\rho_0)v\\[1ex]
   \,=\, &\,\bigl((\partial\omega^{(2)})v\bigr)(R_1,R_2)\,, \qquad
   |\Lambda|\to\infty\,,
\edmal
uniformly on bounded subsets of $(\R^3)^2$, where $\partial\omega^{(2)}$
is the Fr\'echet derivative of $\omega^{(2)}$ in $\L(\V,L^\infty((\R^3)^2)$.
Choosing any fixed box $\Lambda'\subset\Lambda$ we obtain from 
\req{omdiff-integrability} that
\bdm
   \int_{\Lambda'}\,
      \Bigl|
         \omtilde_\Lambda^{(2)}(R,0)-\omega_\Lambda^{(2)}(R,0)
         \,-\, \bigl((\partial\omega_\Lambda^{(2)})v\bigr)(R,0)
      \Bigr|\dR
   \,\leq\, C' \norm{v}_{\V}^2\,,
\edm
and by letting $|\Lambda|\to\infty$ this implies that
\bdm
   \int_{\Lambda'}
   \Bigl|
      \omtilde^{(2)}(R,0)-\omega^{(2)}(R,0)
      \,-\, \bigl((\partial\omega^{(2)})v\bigr)(R,0)
   \Bigr|\dR
   \,\leq\, C' \norm{v}_{\V}^2\,.
\edm
Since the box $\Lambda'\subset\R^3$ can be arbitrarily large
we thus have proved that 
$\partial\omega^{(2)}(\,\cdot\,,0)$ is also the Fr\'echet derivative
of $\omega^{(2)}(\,\cdot\,,0)$ in $\L\bigl(\V,L^1(\R^3)\bigr)$. 
\end{proof}

\section{Integral operator representations of 
$\boldsymbol{\partial\rho^{(1)}}$ and
$\boldsymbol{\partial\rho^{(2)}}$}
\label{Sec:operator}
In a finite size box $\Lambda\subset\R^3$, and 
with $u$ satisfying Assumption~\ref{Ass:u}, the derivatives 
$\partial\rho_\Lambda^{(m)}$ can be represented as integral operators acting
on $\V$. For $m=1,2$ these operators have been computed in 
\cite{Hank16a}, see also Lyubartsev and Laaksonen~\cite{LyLa95}: There holds
\begin{align}
\nonumber
   \bigl((\partial \rho_{\Lambda}^{(1)})v\bigr)(R_1)
   \,=\, &\,-\!\beta \int_\Lambda v(|R_1-R'|)\rho_{\Lambda}^{(2)}(R_1,R')\dR' \\[1ex]
\label{eq:drhodU1}
         &\,-\, \frac{\beta}{2}
                \int_\Lambda\int_\Lambda 
                   v(|R_1'-R_2'|) \chi_\Lambda^{(3)}(R_1,R_1',R_2')
                \dR_1'\!\dR_2'
\end{align}
with $\chi_\Lambda^{(3)}$ of \req{chi3},
and
\begin{align}
\nonumber
   \bigl((\partial \rho_\Lambda^{(2)})v\bigr)&(R_1,R_2)
   \,=\, - \beta\, v(|R_1-R_2|) \rho_\Lambda^{(2)}(R_1,R_2)\\[1ex]
\nonumber
        &- \beta
           \int_\Lambda v(|R_1-R'|)\rho_\Lambda^{(3)}(R_1,R_2,R')\dR' \\[1ex]
\nonumber
        &- \beta
           \int_\Lambda v(|R_2-R'|)\rho_\Lambda^{(3)}(R_1,R_2,R')\dR' \\[1ex]
\label{eq:drhodU2}
        &- \frac{\beta}{2}
           \int_\Lambda\int_\Lambda 
              v(|R_1'-R_2'|) \chi_\Lambda^{(4)}(R_1,R_2,R_1',R_2')
           \dR_1'\!\dR_2'
\end{align}
with $\chi_\Lambda^{(4)}$ of \req{chi4}.
We refer to \cite{LyLa95} and \cite{Hank16a} for physical interpretations 
of these representations.

The goal of this section is to show that the corresponding formulae 
for $\partial\rho^{(m)}$, $m=1,2$, are obtained by integrating over $\R^3$
instead, and by dropping all subscripts $\Lambda$, where
\begin{align*}
   \chi^{(3)}(R_1,R_2,R_3)
   &\,=\, \rho^{(3)}(R_1,R_2,R_3)
          \,-\, \rho^{(1)}(R_1)\rho^{(2)}(R_2,R_3)
\intertext{and}
   \chi^{(4)}(R_1,R_2,R_3,R_4)
   &\,=\, \rho^{(4)}(R_1,R_2,R_3,R_4)
           \,-\, \rho^{(2)}(R_1,R_2)\rho^{(2)}(R_3,R_4)\,,
\end{align*}

Concerning the verification of this assertion for the single integrals 
appearing in \req{drhodU1} and \req{drhodU2}
we utilize the following auxiliary result.

\begin{lemma}
\label{Lem:L1single}
Let $u$ satisfy Assumption~\ref{Ass:u}, 
and for some $R_0\in\R^3$ and $C>0$ let $\xi_\Lambda:\Lambda\to\R$ be a family of 
functions with
\be{C1}
   |\xi_\Lambda(R)| \,\leq\, Ce^{-\beta u(|R-R_0|)}\,, \qquad R\in\Lambda\,,
\ee
independent of $\Lambda$.
Moreover, let $\xi_\Lambda$ converge compactly to $\xi:\R^3\to\R$ as 
$|\Lambda|\to\infty$. Then for every $v\in\V$ there holds
\bdm
   \int_\Lambda v(|R-R_0|)\xi_\Lambda(R)\!\dR 
   \,\to\, \int_{\R^3} v(|R-R_0|)\xi(R)\!\dR
\edm
as $|\Lambda|\to\infty$.
\end{lemma}

\begin{proof}
We extend $\xi_\Lambda$ by zero to $\R^3\setminus\Lambda$ and rewrite
\bdm
   \ell_\Lambda(v) 
   :=\int_\Lambda v(|R-R_0|)\,\xi_\Lambda(R)\!\dR
   = \int_{\R^3} 
        v(|R-R_0|)e^{-\beta u(|R-R_0|)}
        \bigl(\xi_\Lambda(R)e^{\beta u(|R-R_0|)}\bigr)\!
     \dR\,.
\edm
By virtue of \cite[Lemma~3.1]{Hank16a} $\V$ is continuously embedded into
the space 
$\Y_u$ of functions $v:\R^+\to\R$, for which the corresponding norm
\be{normY}
   \norm{v}_{\Y_u}\,:=\,\int_{\R^3} v(|R|)e^{-\beta u(|R|)}\dR
\ee
is finite.
In view of \req{C1} $\ell_\Lambda$ is a linear functional in $\Y_u'$, and 
$\{\ell_\Lambda\}_\Lambda\subset\Y_u'$ is uniformly bounded.
Furthermore, for $v\in\Y_u$ with compact support the compact convergence of 
$\xi_\Lambda\to\xi$ as $|\Lambda|\to\infty$ implies that 
\bdm
   \ell_\Lambda(v) \,\to\, \ell(v) \,=\, \int_{\R^3} v(|R-R_0|)\xi(R)\!\dR\,,
   \qquad |\Lambda|\to\infty\,,
\edm
hence, the assertion of the lemma follows from the Banach-Steinhaus theorem
for every $v\in\Y_u$, and hence, every $v\in\V$.
\end{proof}

To apply this result to \req{drhodU1} and \req{drhodU2} we need to estimate
the molecular distribution functions when their arguments get close.

\begin{proposition}
\label{Prop:rho-near-zero}
Let $u$ be a stable and regular pair potential and let $z$ satisfy \req{z}.
Then there exists $C>0$, independent of the size of $\Lambda$, such that
\be{rho-bound}
   \rho_\Lambda^{(m)}(\RR_m) \,\leq\, 
   C (ze^{2\beta B+1})^m \prod_{i=1}^{m-1} e^{-\beta u(|R_i-R_m|)}
\ee
for all $\RR_m\in\Lambda^m$ and all $m\geq 2$.
\end{proposition}

\begin{proof}
%
With $\I_m=\{1,\dots,m\}$, $\J_{m,N}=\{m+1,\dots,N\}$, and $i^*=m$
we conclude from \req{rhom-graph}, \req{varphi-Ruelle}, 
\req{Ruelle-recursion}, and \req{Ruelle-induction} that
\bdmal
   \bigl|\rho_\Lambda^{(m)}(\RR_m)\bigr|
   &\,\leq\,\sum_{N=m}^\infty \frac{z^N}{(N-m)!}(e^{2\beta B+1})^{N-2}\,
            d_{\I_m,m}(\RR_m)\, \cdot\\
   & \qquad \qquad
            \sum_{\K\subset\J_{m,N}} \!(N-m-|\K|)!\, \cbeta^{N-m-|\K|}
            \int_{\Lambda^{|\K|}}|k_\K(R_m;\RR_\K)|\!\dRR_\K \\[1ex]
   &\,\leq\,\sum_{N=m}^\infty \frac{z^N}{(N-m)!}(e^{2\beta B+1})^{N-2}\,
            d_{\I_m,m}(\RR_m) \cbeta^{N-m}\!\! 
            \sum_{\K\subset\J_{m,N}}\!\!(N-m-|\K|)!\,.
\edmal
The inner sum only depends on the number $p$ of elements in $\K$, 
$0\leq p\leq N-m$, hence
\be{rhom-bound}
\begin{aligned}
   \bigl|\rho_\Lambda^{(m)}(\RR_m)\bigr|
   &\,\leq\,\sum_{N=m}^\infty z^N(e^{2\beta B+1})^{N-2}\,
               d_{\I_m,m}(\RR_m)\cbeta^{N-m}
            \sum_{p=0}^{N-m} \frac{1}{p!} \\[1ex]
   &\,\leq\, \frac{ez^2}{1-z\cbeta e^{2\beta B+1}}\,(ze^{2\beta B+1})^{m-2}
             d_{\I_m,m}(\RR_m)\,. 
\end{aligned}
\ee
The assertion now follows by inserting the definition \req{dm} of $d_{\I_m,m}$.
\end{proof}

Combining Lemma~\ref{Lem:L1single} and Proposition~\ref{Prop:rho-near-zero}
we readily obtain the thermodynamical limits of the three single integrals
occurring in \req{drhodU1} and \req{drhodU2}. 
The thermodynamical limits of the remaining two double integrals 
involving $\chi_\Lambda^{(m)}$, $m=3,4$ in \req{drhodU1} and \req{drhodU2},
respectively, are more subtle and will be considered next.

Again, we start with an auxiliary result.

\begin{lemma}
\label{Lem:L1double}
Let $u$ satisfy Assumption~\ref{Ass:u}, and $\chi_\Lambda:\Lambda^2\to\R$ 
be a family of functions with
\be{fbar}
   |\chi_\Lambda(R_1,R_2)| 
   \,\leq\, e^{-\beta u(|R_1-R_2|)} \bigl(X(R_1)+X(R_2)\bigr)
\ee
for all $R_1,R_2\in\Lambda$,
where $X\in L^1(\R^3)$ is nonnegative and bounded and
does not depend on $\Lambda$. 
Furthermore, assume that $\chi_\Lambda$ converges compactly
to $\chi:(\R^3)^2\to\R$ as $|\Lambda|\to\infty$. Then for every
$v\in\V$ there holds
\bdm
   \int_\Lambda\int_\Lambda v(|R_1-R_2|)\chi_\Lambda(R_1,R_2)\dR_1\!\dR_2
   \,\to\, \int_{\R^3} \int_{\R^3} v(|R_1-R_2|)\chi(R_1,R_2)\dR_1\!\dR_2
\edm
as $|\Lambda|\to\infty$.
\end{lemma}

\begin{proof}
Throughout, we extend $\chi_\Lambda$ by zero to $(\R^3)^2\setminus\Lambda^2$,
and this extension, of course, satisfies \req{fbar} for all 
$R_1,R_2\in\R^3$. Because of the compact convergence $\chi_\Lambda\to \chi$ as
$|\Lambda|\to\infty$ this inequality also extends to $\chi$, i.e., 
\be{fbarf}
   |\chi(R_1,R_2)| \,\leq\, e^{-\beta u(|R_1-R_2|)} \bigl(X(R_1)+X(R_2)\bigr)\,,
   \qquad
   R_1,R_2\in\R^3\,.
\ee 

Substituting $R_1'=R_1-R_2$ and $R_2'=R_1+R_2$ we obtain
\bdmal
   &\int_\Lambda\int_\Lambda v(|R_1-R_2|) \chi_\Lambda(R_1,R_2)\dR_1\!\dR_2\\[1ex]
   &\qquad
    \,=\, \frac{1}{8}\int_{\R^3} v(|R_1'|)e^{-\beta u(|R_1'|)}
             \int_{\R^3} e^{\beta u(|R_1'|)}
                \chi_\Lambda\bigl(\tfrac{R_1'+R_2'}{2},
                                  \tfrac{R_2'-R_1'}{2}\bigr)
             \dR_2'\!\dR_1'\,,
\edmal
and hence, since $v\in \V\subset\Y_u$, compare~\req{normY},
it only remains to show that
\bdm
   J(R_1')
   \,=\, \int_{\R^3} e^{\beta u(|R_1'|)}
            \Bigl(
               \chi_\Lambda\bigl(\tfrac{R_1'+R_2'}{2},
                                 \tfrac{R_2'-R_1'}{2}\bigr)
               \,-\, \chi\bigl(\tfrac{R_1'+R_2'}{2},\tfrac{R_2'-R_1'}{2}\bigr)
            \Bigr) \dR_2'
\edm
is uniformly bounded for $R_1'\in\R^3$, and converges compactly to zero.
The uniform boundedness follows readily from \req{fbar} and \req{fbarf},
since
\bdm
   |J(R_1')|
   \,\leq\, 2\int_{\R^3} X(\tfrac{R_1'+R_2'}{2})\dR_2'
            \,+\, 2\int_{\R^3} X(\tfrac{R_2'-R_1'}{2})\dR_2'
   \,=\, 32\, \norm{X}_{L^1(\R^3)}\,. 
\edm

To prove the compact convergence $J\to 0$ we introduce for $r'>0$ the 
spherical shell $\A_{r'}=\{1/r'\leq|R_2'|\leq r'\}$, and estimate
\bdmal
   |J(R_1')|
   &\,\leq\, \int_{\A_{r'}} e^{\beta u(|R_1'|)}
               \Bigl|
                 \chi_\Lambda\bigl(\tfrac{R_1'+R_2'}{2},
                                   \tfrac{R_2'-R_1'}{2}\bigr)
                 \,-\, \chi\bigl(\tfrac{R_1'+R_2'}{2},
                                 \tfrac{R_2'-R_1'}{2}\bigr)
               \Bigr| \dR_2'\\
   &\qquad\quad
             \,+\, 
             \int_{\R^3\setminus\A_{r'}} e^{\beta u(|R_1'|)}
               \Bigl(
               \bigl|\chi_\Lambda\bigl(\tfrac{R_1'+R_2'}{2},
                                       \tfrac{R_2'-R_1'}{2}\bigr)\bigr|
               \,+\, \bigl|\chi\bigl(\tfrac{R_1'+R_2'}{2},
                                     \tfrac{R_2'-R_1'}{2}\bigr)
                     \bigr|\Bigr) \dR_2'\\[1ex]
   &\,\leq\, \int_{\A_{r'}} e^{\beta u(|R_1'|)}
               \Bigl|
                 \chi_\Lambda\bigl(\tfrac{R_1'+R_2'}{2},
                                   \tfrac{R_2'-R_1'}{2}\bigr)
                 \,-\,\chi\bigl(\tfrac{R_1'+R_2'}{2},\tfrac{R_2'-R_1'}{2}\bigr)
               \Bigr| \dR_2'\\
   &\qquad\quad
             \,+\, 
             2 \int_{\R^3\setminus\A_{r'}}\!\!
                   \Bigl(X(\tfrac{R_1'+R_2'}{2}) 
                         + X(\tfrac{R_2'-R_1'}{2})\Bigr) \dR_2'\\[1ex]
   &\,\leq\, \int_{\A_{r'}} e^{\beta u(|R_1'|)}
               \Bigl|\chi_\Lambda\bigl(\tfrac{R_1'+R_2'}{2},
                                       \tfrac{R_2'-R_1'}{2}\bigr)
                 \,-\,\chi\bigl(\tfrac{R_1'+R_2'}{2},\tfrac{R_2'-R_1'}{2}\bigr)
               \Bigr| \dR_2'\\
   &\qquad\quad
             \,+\, 32 \int_{|R|>r} X(R)\dR
             \,+\, 16 \int_{\B_+} X(R)\dR \,+\, 16 \int_{\B_-} X(R)\dR\,,
\edmal
where $r=(r' - |R_1'|)/2$
and $\B_\pm=\{R:|R\pm R_1'/2|<1/(2r')\}$.
Given a compact set $\Omega\subset\R^3$ and any $\eps>0$ we can fix $r'$ 
so large that the sum of the latter three integrals is 
bounded by $\eps/2$ for every $R_1'\in\Omega$. Moreover, for $R_1'\in\Omega$ 
and $R_2'\in\A_{r'}$ we can use the compact convergence 
of $\chi_\Lambda$ to also bound the former integral by $\eps/2$ by choosing 
$|\Lambda|$ sufficiently large. Thus we have shown that
\bdm
   |J(R_1')| \,\leq\, \eps \qquad \text{for all $R_1'\in \Omega$}\,,
\edm
provided that $|\Lambda|$ is sufficiently large. In other words, there holds
$J\to 0$ as $|\Lambda|\to\infty$, uniformly in $\Omega$, which was to be shown.
\end{proof}

This result, together with the estimates of $\chi_\Lambda^{(3)}$ and
$\chi_\Lambda^{(4)}$ in Propositions~\ref{Prop:chi3} and \ref{Prop:chi4},
respectively, shows that the double integrals~\req{drhodU1}
and \req{drhodU2} have a well-defined thermodynamical limit.
In particular, taking into account that the thermodynamical limits of
$\rho^{(m)}$ are even and translation invariant functions, we find that
\be{rhodU2-limit}
\begin{aligned}
   \bigl((\partial \rho^{(2)})v\bigr)(R,0)
   \,=\,&-\beta\, v(|R|) \rho^{(2)}(R,0)
         \,-\, 2\beta
               \int_{\R^3} v(|R'|)\rho^{(3)}(R,0,R')\dR' \\[1ex]
        &- \frac{\beta}{2}
           \int_{\R^3} v(|R'|)
           \int_{\R^3}
              \chi^{(4)}(R,0,R'',R''+R')
           \dR''\!\dR'\,,
\end{aligned}
\ee
with all integrals converging absolutely.


\end{document}